\documentclass[conference]{IEEEtran}
\IEEEoverridecommandlockouts
\usepackage{inputenc}
\usepackage{placeins}
\usepackage[american]{babel}
\usepackage{csquotes}
\usepackage[letterpaper,left=48pt,bottom=43pt,right=48pt,top=57pt]{geometry}

\usepackage{comment}
\usepackage[listings,light]{solarized}
\usepackage{tikz}
\usepackage{pgfmath}
\usetikzlibrary{babel}
\usetikzlibrary{intersections}
\usetikzlibrary{fadings}
\usetikzlibrary{calc}
\usetikzlibrary{math}
\usetikzlibrary{fit}
\usetikzlibrary{positioning}
\usetikzlibrary{patterns}
\usetikzlibrary{shapes, arrows, arrows.meta}
\usetikzlibrary{shapes.geometric, shapes.multipart}
\usetikzlibrary{decorations.pathreplacing}
\usetikzlibrary{decorations.pathmorphing}
\usetikzlibrary{decorations.text}
\usetikzlibrary{backgrounds}
\pgfdeclarelayer{background}
\pgfdeclarelayer{foreground}
\pgfsetlayers{background,main,foreground}%,foreforeground}
\usepackage{aircraftshapes}
\usepackage{url}
\usepackage{hyperref}
\hypersetup{colorlinks=false,linkcolor=blue,urlcolor=blue}
\usepackage{graphicx}
\usepackage{subcaption}
\usepackage{tikz}
\usetikzlibrary{arrows.meta,calc,decorations.markings}
\addto\captionsenglish{\renewcommand{\figurename}{Fig.}}
\addto\captionsamerican{\renewcommand{\figurename}{Fig.}}
\usepackage{multicol}
\usepackage{xspace}
\usepackage{mathtools}
\usepackage{booktabs}
\usepackage{multirow}
\usepackage{tabularx}
\newcolumntype{C}{>{\centering\arraybackslash}X}
\newcolumntype{L}{>{\raggedright\arraybackslash}X}
\usepackage{siunitx}
\usepackage{listings}
\usepackage{float}
\usepackage{bm}         % Negritas en modo matemático
\usepackage{amsmath,amssymb,amsfonts,amsthm}

\newtheorem{theorem}{Theorem}
\newtheorem{corollary}{Corollary}

\usepackage{bm}

\allowdisplaybreaks
\begin{document}

\title{Robust Geometric Control of Catenary Robots under Unstructured Force Uncertainties\\
\thanks{The authors acknowledge financial support from Grant PID2022-137909NB-C21 funded by MCIN/AEI/ 10.13039/501100011033. A.A.S. was partially supported by MICIU/AEI/10.13039/501100011033/ FEDER, UE, Grant No. PID2024-155187OB-I00. L.C was supported in part by iRoboCity2030-CM, Robótica Inteligente para Ciudades Sostenibles (TEC-2024/TEC-62), funded by the Programas de Actividades I+D en Tecnologías en la Comunidad de Madrid.}}

\author{\IEEEauthorblockN{Alexandre Anahory Sim\~oes}
\IEEEauthorblockA{\textit{IE School of Science and Technology,} \\
\textit{IE University}\\
Madrid, Spain \\
alexandre.anahory@ie.edu}
\and
\IEEEauthorblockN{Leonardo Colombo}
\IEEEauthorblockA{\textit{Centre for Automation and Robotics (CSIC-UPM)} \\
Madrid, Spain\\
leonardo.colombo@csic.es}
}

\maketitle
% ====================================================
% --- Added in draft --- ABSTRACT --------------------
\begin{abstract}
This paper considers the robust control of a catenary robot composed of
two quadrotors connected by an inextensible cable. The system is modeled
on \(SE(3)\), with the cable treated as a geometric subsystem induced by
the UAV configuration rather than as an independent dynamical element.
The catenary shape determines configuration-dependent forces that couple
the translational dynamics of the vehicles. We propose a geometric
tracking controller for the relative configuration of the agents and
analyze its robustness with respect to unstructured uncertainties in the
catenary-induced forces. The main theoretical result establishes local
input-to-state stability of the closed-loop tracking errors. In
particular, we obtain asymptotic convergence in the nominal case and an
explicit ultimate bound for the tracking errors under bounded
catenary-force perturbations.
\end{abstract}

% --- End added in draft -----------------------------

% \begin{IEEEkeywords}
% Aerial manipulation, catenary robots, geometric control, cable-driven
% robots, quadrotors, geometric mechanics.
% \end{IEEEkeywords}

% ====================================================
\section{Introduction}\label{sec:intro}

Cable-connected aerial systems have attracted attention due
to their ability to coordinate multiple vehicles through lightweight and
compliant mechanical links. Compared with rigid or articulated aerial
mechanisms, cable-based systems offer reduced mechanical complexity,
favorable payload-to-weight ratios, and intrinsic compliance \cite{Elastic,Lee-TCST,LeeSrePICDC13}. These
features make them useful in cooperative transportation, load
interaction, and contact-rich aerial tasks where force transmission
through the cable plays a central role \cite{goodman2023geometric, jacob2, pacheco1, pacheco2}.

A particular instance of this class is the \emph{catenary robot},
consisting of two UAVs connected by an inextensible cable \cite{d2021catenary, cardona2021adaptive}.
When gravity is taken into account, the cable assumes a catenary shape
whose geometry depends on the relative configuration of the vehicles.
Thus, the cable does not merely impose a distance constraint: it induces
configuration-dependent forces that couple the motion of the aerial
agents. This coupling can be exploited to regulate the
relative configuration of the vehicles and, indirectly, the shape and
orientation of the cable.

Most existing approaches to cable-connected aerial robots rely on
task-specific models, planar reductions, or simplified descriptions of
the cable forces \cite{jackson2020scalable, oliva2024aerial}. In many cases, the cable is treated either as a
kinematic constraint or as a nominal force-transmission element. Such
models are effective for specific scenarios, but they do not provide a
systematic geometric description of the relation between the UAV
configuration, the induced catenary shape, and the resulting endpoint
forces. This also makes robustness analysis difficult when the ideal
catenary model is affected by unmodeled effects such as cable inertia,
elasticity, aerodynamic drag, or transient oscillations.

This paper addresses this issue by developing a geometric modeling and
control framework for catenary robots. We show
that, on the admissible configuration set, the relative position of the
two aerial agents determines the catenary through a finite set of shape
parameters: its orientation, span, and catenary parameter. The cable is
therefore treated as a geometric subsystem induced by the UAV
configuration, rather than as an independent dynamical element. The
corresponding endpoint tensions define smooth configuration-dependent
forces that enter the translational dynamics of the vehicles.

Building on this structure, we design a geometric tracking controller
for the relative configuration of the two vehicles. The nominal
controller compensates the catenary-induced forces as feedforward terms,
while the attitude dynamics are handled through a standard geometric
controller on \(SO(3)\). We then model deviations from the ideal
quasi-static catenary description as unstructured additive uncertainties
in the catenary forces. We prove local input-to-state
stability of the closed-loop tracking errors with respect to these
uncertainties. As consequences, the nominal tracking errors converge to
zero, while bounded catenary-force perturbations lead to an explicit
ultimate bound whose size decreases with the damping gains.

The remainder of the paper is organized as follows. Section~II presents
the geometric model of the catenary robot. Section~III derives the
catenary-induced forces and formulates the control objective. Section~IV
develops the robust geometric controller and provides the numerical
illustration.

% ====================================================

\section{Geometric Modeling of Catenary Robots}

We consider two multi-rotor UAVs evolving on the special Euclidean group
$SE(3)$, representing position and attitude in three-dimensional space. The configuration of UAV $i \in \{1,2\}$ is
$g_i = (R_i,p_i) \in SE(3)$, where $R_i \in SO(3)$ denotes the attitude and
$p_i \in \mathbb{R}^3$ the position expressed in an inertial frame. We will make use of the following reference frames:

\begin{itemize}
\item \emph{Inertial frame} $\{I\}$: a fixed world frame with origin
$O_I$ and orthonormal basis $\{e_1,e_2,e_3\}$, where $e_3$ is aligned with
the direction opposite to gravity. All positions $p_i$ and velocities
$v_i$ are expressed in this frame.

\item \emph{Body-fixed frames} $\{B_i\}$: a frame rigidly attached to each
aerial agent $i \in \{1,2\}$, with origin at the vehicle's center of mass.
The rotation matrix $R_i \in SO(3)$ maps vectors expressed in $\{B_i\}$ to
the inertial frame $\{I\}$. In particular, the body-fixed thrust direction
corresponds to the third basis vector $e_3$ of $\{B_i\}$.

\item \emph{Catenary frame} $\{C\}$: a moving frame associated with the
plane containing the catenary. The frame $\{C\}$ is defined such that its
third axis is aligned with the inertial vertical direction $e_3$, and its
second axis lies along the horizontal direction connecting the projections
of the agents' positions onto the horizontal plane. The orientation of
$\{C\}$ with respect to $\{I\}$ is given by
$R_C \in SO(3)$.
\end{itemize}

A schematic representation of the catenary robot configuration and the reference frames given above is shown in Fig.~\ref{fig:catenary_geometry}.

%Unless otherwise specified, vectors expressed in $\mathbb{R}^3$ are
%assumed to be represented in the inertial frame $\{I\}$. Quantities
%expressed in the catenary frame $\{C\}$ are explicitly denoted as such.

\begin{figure}[h!]
\centering
\begin{tikzpicture}[scale=1.08, line cap=round, line join=round, >=Latex]

%========================================================
% COLORS
%========================================================
\definecolor{DroneBlue}{RGB}{0,39,76}         % Michigan blue
\definecolor{CatenaryGold}{RGB}{255,203,5}    % Michigan maize
\definecolor{CatenaryGoldDark}{RGB}{210,160,0}
\definecolor{FrameRed}{RGB}{180,70,60}
\definecolor{FrameTeal}{RGB}{0,120,120}
\definecolor{PlaneBlue}{RGB}{120,140,255}

\tikzset{
    axisI/.style={->, semithick, black},
    axisCvert/.style={->, semithick, FrameRed},
    axisChorz/.style={->, semithick, FrameTeal},
    force/.style={->, semithick, DroneBlue},
    tension/.style={->, semithick, CatenaryGoldDark},
    helper/.style={dashed, thin, gray!70},
    ground/.style={fill=gray!10, draw=gray!45, line width=0.3pt},
    catplane/.style={fill=PlaneBlue!8, draw=PlaneBlue!35, line width=0.4pt},
    cable/.style={line width=1.2pt, CatenaryGold},
    dronearm/.style={line width=1.0pt, DroneBlue},
    rotor/.style={line width=0.8pt, DroneBlue},
    lab/.style={font=\small},
    slab/.style={font=\scriptsize},
    pt/.style={circle, fill, inner sep=1.0pt},
}

%========================================================
% GROUND PLANE
%========================================================
\coordinate (G1) at (0.55,0.55);
\coordinate (G2) at (5.95,0.55);
\coordinate (G3) at (7.20,1.28);
\coordinate (G4) at (1.80,1.28);
\filldraw[ground] (G1)--(G2)--(G3)--(G4)--cycle;

%========================================================
% INERTIAL FRAME {I}
%========================================================
\coordinate (OI) at (1.05,0.92);
\draw[axisI] (OI) -- ++(1.02,0) node[below right=-1pt] {$e_1$};
\draw[axisI] (OI) -- ++(0.72,0.40) node[right=-1pt] {$e_2$};
\draw[axisI] (OI) -- ++(0,1.36) node[left=-1pt] {$e_3$};
\node[lab, below left=1pt] at (OI) {$\{I\}$};

%========================================================
% UAV POSITIONS / ENDPOINTS
%========================================================
\coordinate (P1) at (2.70,3.82);
\coordinate (P2) at (6.00,3.82);

% Projections on ground plane
\coordinate (P1xy) at (2.85,1.08);
\coordinate (P2xy) at (5.90,1.08);

\draw[helper] (P1) -- (P1xy);
\draw[helper] (P2) -- (P2xy);
\node[pt] at (P1xy) {};
\node[pt] at (P2xy) {};

% Equal-height guide
\draw[dashed, thin, gray!60] (2.20,3.82) -- (6.30,3.82);

%========================================================
% HORIZONTAL DIRECTION d
%========================================================
\draw[axisI] (P1xy) -- (P2xy) node[midway, below=3pt] {$d$};

%========================================================
% YAW ANGLE psi_C
%========================================================
\draw[thin] (1.86,0.92) arc[start angle=0,end angle=14,radius=0.76];
\node[slab] at (2.25,1.13) {$\psi_C$};

%========================================================
% VERTICAL CATENARY PLANE
%========================================================
\coordinate (CPL1) at (2.05,1.50);
\coordinate (CPL2) at (6.45,1.50);
\coordinate (CPL3) at (6.45,4.45);
\coordinate (CPL4) at (2.05,4.45);
\filldraw[catplane] (CPL1)--(CPL2)--(CPL3)--(CPL4)--cycle;

%========================================================
% CATENARY
%========================================================
% Vertex chosen at x = 4.35, y = 2.80
\draw[cable, samples=140, domain=-1.65:1.65, smooth, variable=\x]
plot ({4.35+\x}, {2.80 + 0.42*cosh(\x/0.92) - 0.42});

%========================================================
% CATENARY FRAME {C} with origin at the vertex
%========================================================
\coordinate (OC) at (4.35,2.80);
\draw[axisChorz] (OC) -- ++(0.78,0) node[below=1pt, FrameTeal] {$e_2^{C}$};
\draw[axisCvert] (OC) -- ++(0,0.82) node[left=1pt, FrameRed] {$e_3^{C}$};
\node[lab, below left=1pt] at (OC) {$\{C\}$};

%========================================================
% UAV 1
%========================================================
\begin{scope}[shift={(2.70,3.82)}, rotate=15]
    \draw[dronearm] (-0.24,0) -- (0.24,0);
    \draw[dronearm] (0,-0.16) -- (0,0.16);
    \draw[rotor] (-0.26,0) circle (0.032);
    \draw[rotor] ( 0.26,0) circle (0.032);
    \draw[rotor] (0,-0.18) circle (0.032);
    \draw[rotor] (0, 0.18) circle (0.032);
    \fill[DroneBlue] (0,0) circle (1.1pt);
\end{scope}

%========================================================
% UAV 2
%========================================================
\begin{scope}[shift={(6.00,3.82)}, rotate=-13]
    \draw[dronearm] (-0.24,0) -- (0.24,0);
    \draw[dronearm] (0,-0.16) -- (0,0.16);
    \draw[rotor] (-0.26,0) circle (0.032);
    \draw[rotor] ( 0.26,0) circle (0.032);
    \draw[rotor] (0,-0.18) circle (0.032);
    \draw[rotor] (0, 0.18) circle (0.032);
    \fill[DroneBlue] (0,0) circle (1.1pt);
\end{scope}

%========================================================
% LABELS p1, p2
%========================================================
\node[lab] at (2.33,4.00) {$p_1$};
\node[lab] at (6.37,4.00) {$p_2$};

%========================================================
% THRUSTS
%========================================================
\draw[force] (2.70,3.82) -- ++(-0.10,0.86)
    node[above left=-1pt, DroneBlue] {$f_1R_1e_3$};

\draw[force] (6.00,3.82) -- ++(0.10,0.86)
    node[above right=-1pt, DroneBlue] {$f_2R_2e_3$};

%========================================================
% TENSIONS
%========================================================
\tikzset{tensionLine/.style={thick}}
\draw[tensionLine] (2.70,3.82) -- ++(0.16,-0.50)
    node[left, above=3pt, xshift=3pt, CatenaryGoldDark] {$T_1$};

\draw[tensionLine] (6.00,3.82) -- ++(-0.16,-0.50)
    node[right, above=2pt, xshift=-5pt, CatenaryGoldDark] {$T_2$};
%\draw[tension] (2.70,3.82) -- ++(0.56,-0.10) % 'tension' is a TikZ style, not a key
   % node[midway, above=2pt, xshift=1pt, CatenaryGoldDark] {$T_1$};

%\draw[tension] (6.00,3.82) -- ++(-0.56,-0.10)
 %   node[midway, above=2pt, xshift=-1pt, CatenaryGoldDark] {$T_2$};

\end{tikzpicture}
\caption{Configuration of the catenary robot.}
\label{fig:catenary_geometry}
\end{figure}
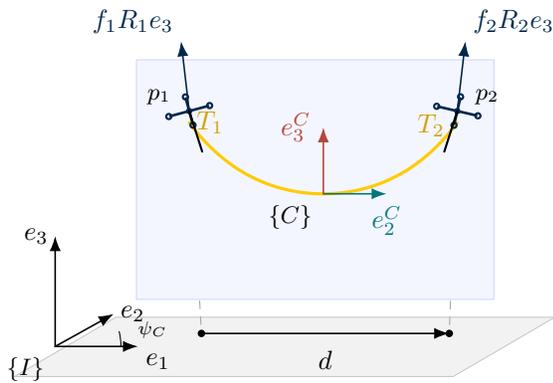

The body angular velocity $\Omega_i \in \mathbb{R}^3$ and translational
velocity $v_i \in \mathbb{R}^3$ satisfy the kinematics
\begin{equation}
\dot{R}_i = R_i \hat{\Omega}_i, \qquad \dot{p}_i = v_i,
\end{equation}
where $(\cdot)^\wedge : \mathbb{R}^3 \to \mathfrak{so}(3)$ denotes the hat
map (i.e., and isomorphism between vectors on $\mathbb{R}^{3}$ with skew-symmetric matrices, see e.g., \cite{lee2018global}).

Each UAV is an underactuated aerial robot of quadrotor type, with
control inputs $u_i = (f_i,\tau_i) \in \mathbb{R} \times \mathbb{R}^3$,
where $f_i$ is the total thrust generated along the body-fixed $e_3$
axis and $\tau_i \in \mathbb{R}^3$ are the control torques. The equations
of motion of UAV $i$ are 
\begin{align}
\dot{p}_i &= v_i, \label{eq:transl_kin}\\
m_i \dot{v}_i &= f_i R_i e_3 - m_i g e_3 + R_C T_i, \label{eq:transl_dyn}\\
\dot{R}_i &= R_i \hat{\Omega}_i, \label{eq:rot_kin}\\
J_i \dot{\Omega}_i + \Omega_i \times J_i \Omega_i &= \tau_i,
\label{eq:rot_dyn}
\end{align}
where $m_i>0$ is the mass, $J_i \in \mathbb{R}^{3\times 3}$ is the inertia
matrix, $g>0$ is the gravitational constant, and $T_i \in \mathbb{R}^3$
denotes the tension exerted by the cable on UAV $i$, in the
catenary frame.

The UAVs are connected by a massless, inextensible cable of fixed length
$l>0$. Consequently, the system is subject to the holonomic constraint
$
0 < \|p_1 - p_2\| \leq l$, which prevents both collision between the agents and cable slackness. Throughout the paper, we assume that the system evolves on the admissible
configuration set
\begin{equation*}
\mathcal{A}
=
\left\{
(p_1,p_2) \in \mathbb{R}^3 \times \mathbb{R}^3
\ \middle| \
0 < \|p_1 - p_2\| \leq l, \ p_{1z}=p_{2z}
\right\}.
\end{equation*}

Each admissible configuration of the pair $(p_1,p_2)$ uniquely determines
the configuration of a catenary. We will also assume that the catenary is in standard configuration, i.e., its endpoints are located at the same height, which imposes the second constraint in $\mathcal{A}$ by requiring that $p_{1z}=p_{2z}$.

We will describe the catenary configuration by the tuple
$c = (R_C, s, a) \in SO(3) \times \mathbb{R}_{>0} \times \mathbb{R}_{>0}$, where $R_C$ denotes the orientation of the vertical plane containing the
catenary, $s$ is the half-span, and $a$ is the catenary shape parameter. Equivalently, the catenary configuration may be viewed as the minimizer of a potential energy
functional under fixed endpoint constraints. 

The orientation $R_C$ is defined such that the catenary lies in the
vertical plane spanned by $e_3$ and the horizontal direction
\[
d = \frac{(p_2 - p_1)_{xy}}{\|(p_2 - p_1)_{xy}\|},
\]
where $(\cdot)_{xy}$ denotes projection onto the horizontal plane.
Accordingly, we define
\begin{equation}
R_C = R_z(\psi_C), \qquad
\psi_C = \mathrm{atan2}(y_2 - y_1, x_2 - x_1).
\end{equation}

The projection $(\cdot)_{xy}$ is used to extract the horizontal separation
between the UAVs, which uniquely defines the vertical plane containing
the catenary together with the gravity direction $e_3$. The half-span of the catenary is given by
\begin{equation}
s = \frac{\|(p_2 - p_1)_{xy}\|}{2}.
\end{equation}
The shape parameter $a>0$ is uniquely determined by the cable length $l$
through the intrinsic catenary relation
\begin{equation}
\frac{l}{2} = a \sinh\!\left(\frac{s}{a}\right).
\end{equation}

Let $r \in [-s,s]$ be the catenary parameter. In the catenary frame
$\{C\}$, the cable centerline is described by
\begin{equation}
\alpha_C(r) =
\begin{bmatrix}
0 \\
r \\
a\!\left(\cosh\!\frac{r}{a} - 1\right)
\end{bmatrix}.
\end{equation}
Note that this expression is valid whenever the endpoints are at the same height.

The catenary curve expressed in the inertial frame is then given by
$\alpha(r) = R_C \alpha_C(r) + c_0$, where the offset
\begin{equation}
c_0 = \tfrac{1}{2}(p_1+p_2) -
a\!\left(\cosh\!\frac{s}{a}-1\right)e_3
\end{equation}
corresponds to the inertial frame coordinates of the catenary vertex and ensures that the endpoints satisfy $\alpha(-s)=p_1$ and
$\alpha(s)=p_2$. In addition, for all $(p_1,p_2) \in \mathcal{A}$, the catenary configuration and the
corresponding endpoint tensions $T_1,T_2$ are uniquely defined and depend
smoothly on $(p_1,p_2)$.

\section{Catenary-Induced Forces and Control Objective}

In this section, we characterize the forces exerted by the catenary on the UAVs as a consequence of the cable geometry and gravity and the control objective. Since the cable is assumed to be inextensible and its inertia is neglected,
its configuration is not described by independent dynamical states. Instead, the cable is assumed to remain in quasi-static equilibrium under its own distributed weight and the endpoint constraints imposed by the quadrotors.

\subsection{Catenary-Induced Forces}

Consider the catenary curve parameterized by $r \in [-s,s]$ in the frame $\{C\}$. Under gravity, an inextensible cable suspended between two points assumes
a catenary shape determined by static equilibrium. In this configuration,
the tension along the cable has a constant horizontal component, while
its vertical component varies along the arc length to balance gravity.

In our model, the cable is assumed to be inextensible and to
remain in quasi-static equilibrium under gravity. Its configuration is
therefore fully determined by the positions of the aerial vehicles, and it
does not introduce additional independent dynamical states. Here \(w>0\) denotes the effective weight per unit length of the cable.
Thus, although the cable inertia is neglected in the equations of motion,
its gravitational loading is retained in the static equilibrium problem that
determines the catenary shape and the endpoint tensions.

Under these assumptions, the tension vector expressed in the catenary
frame at parameter $r$ is given by
\begin{equation}
T_C(r)=
\begin{bmatrix}
0\\
\pm wa\\
wa\,\sinh\!\left(\frac{r}{a}\right)
\end{bmatrix},
\label{eq:catenary_tension_distribution}
\end{equation}
where the sign of the horizontal component depends on the orientation
along the cable. The signs of the horizontal
components are chosen so that the endpoint forces act from
each attachment point toward the interior of the cable. Hence,
the endpoint forces have opposite horizontal components and
equal vertical components. Its magnitude is $\|T_C(r)\| = wa\,\cosh\!\left(\frac{r}{a}\right)$, which increases monotonically with $|r|$ and attains its maximum at the
endpoints of the catenary.

The forces exerted by the cable on the UAVs correspond to the
endpoint tensions evaluated at $r=\pm s$. Denoting by $T_1$ and $T_2$ the
tensions acting on UAVs $1$ and $2$, respectively, and expressed in the
catenary frame $\{C\}$, we obtain $T_1 = T_C(-s)$ and 
$T_2 = T_C(s)$.

These forces are mapped into the inertial frame through the catenary
orientation $R_C$, yielding the contribution $R_C T_i$ in the translational
dynamics of each agent, as shown in \eqref{eq:transl_dyn}.

By symmetry, the horizontal components of the endpoint tensions are equal
in magnitude and opposite in direction, while the tension magnitude is the
same at both attachment points. As a result, the catenary induces both
coupling forces between the UAVs and an additional load that must be
balanced by the thrust inputs.

A key observation is that the endpoint tensions $T_i$ depend on the agent
configuration solely through the geometric parameters of the catenary,
namely the orientation $R_C$, the half-span $s$, and the shape parameter
$a$. Since these quantities are uniquely determined by the relative
positions of the quadrotors, the catenary-induced forces can be interpreted
as a smooth map $(p_1,p_2)\longmapsto (T_1,T_2)$, which couples the translational dynamics of the aerial vehicles.

This geometric dependence highlights that the cable does not introduce
additional dynamical states, but acts as a shape-dependent force
transmission mechanism. Consequently, control of the UAVs implicitly
shapes the catenary configuration and modulates the forces exchanged
through the cable.

The smoothness statement holds on the open admissible set corresponding to
non-degenerate catenary configurations, namely configurations for which the
horizontal endpoint separation is strictly smaller than the cable length.
The limiting case in which the cable is completely taut is excluded from
the present analysis.

\subsection{Control objective}\label{controlobj}

The control inputs available for each UAV are the thrust magnitude
$f_i \in \mathbb{R}$ and the body torques $\tau_i \in \mathbb{R}^3$, which
directly affect the translational and rotational dynamics of the
quadrotors. Through these inputs, the UAVs can influence their positions
$p_i$ and orientations $R_i$, which in turn uniquely determine the
configuration of the catenary and the forces transmitted through the
cable. The cable does not introduce additional dynamical
states, but induces configuration-dependent forces $T_i$ that couple the
motion of the UAVs. Therefore, the control problem consists in designing
the inputs $(f_i,\tau_i)$ so as to achieve desired cooperative behaviors
while respecting the geometric constraints imposed by the cable.

%\subsection{Control Objectives and Problem Statement}

The primary control objective is to regulate the relative
configuration of the UAVs through the catenary-induced coupling.
Specifically, given a desired relative position trajectory
$p_{12}^d(t) \in \mathbb{R}^3$, the goal is to drive the system such that
\begin{equation}
p_2(t) - p_1(t) \rightarrow p_{12}^d(t),
\end{equation}
while maintaining the feasibility of the cable constraint.

Equivalently, this objective can be interpreted as driving the catenary
configuration toward a prescribed geometric shape
\begin{equation}
c(t) = (R_C(t), s(t), a(t)) \rightarrow c_d(t),
\end{equation}
where $c_d(t)$ denotes a desired time-varying catenary configuration
consistent with the cable length constraint.

%\textcolor{red}{Note that by defining a desired relative position, we are not imposing a unique trajectory in space for the catenary. For instance, the controller would not distinguish between two drones standing still or moving in parallel. In both cases the desired relative position would be constant. What we are controlling in this  way is the geometric shape and orientation of the catenary and not its dynamics.}

This objective is invariant under common translations of the
two UAVs. Therefore, prescribing \(p_{12}^d(t)\) does not define
a unique absolute trajectory for the catenary in the inertial
frame. Rather, it specifies the relative geometry of the UAV
pair, and consequently the orientation and shape parameters
\((R_C,s,a)\) of the induced catenary. In this sense, the
controller regulates the catenary geometry through the UAV
relative configuration, rather than assigning independent
dynamics to the cable.

Beyond the relative position objective, additional performance goals include
the alignment of each agent attitude $R_i$ with the required force directions,
the preservation of bounded translational and rotational velocities, and the
smooth shaping of catenary-induced forces to prevent excessive tension peaks.

%The control problem addressed in this paper can be summarized as follows.

\medskip
\noindent\textbf{Problem 1.}
\emph{Given a feasible desired relative trajectory
\(p_{12}^d(t)\), design thrust and torque inputs
\((f_i,\tau_i)\), \(i=1,2\), such that the relative tracking error
$e_p(t) = (p_2(t)-p_1(t)) - p_{12}^d(t)$ converges to zero in the nominal catenary model, while the
closed-loop trajectories remain in the admissible set
\(\mathcal A\). Moreover, when the nominal catenary-induced
forces are affected by bounded unstructured uncertainties,
the closed-loop relative tracking error should remain bounded
with an ultimate bound depending on the size of the uncertainty.}

%\textcolor{red}{Let us read again this section after we finish the paper. This section seems a description of what we will do.}

\section{Geometric Control and Robustness Analysis}

In this section, we develop a geometric control strategy for the
catenary robot that achieves the objectives stated in Section~\ref{controlobj}. %The proposed controller is designed directly on $SE(3)$ and explicitly exploits the geometric
%structure of the catenary-induced coupling between the aerial robots.
%We first present the nominal tracking controller and then analyze
%its robustness with respect to uncertainties in the catenary model.

\subsection{Geometric tracking control}

To stabilize the desired relative trajectory of the two-UAV formation,
we assign an individual reference trajectory to each UAV. In particular,
we stabilize UAV~1 around the constant point $p_1^d(t)\equiv p_{1,0}\in\mathbb{R}^3$, which may be chosen, for instance, as its initial location, and we define
the desired trajectory of UAV~2 by $p_2^d(t)=p_{1,0}+p_{12}^d(t)$. In this way,
\[
p_2^d(t)-p_1^d(t)=p_{12}^d(t),
\]
so that tracking of the individual reference trajectories implies
convergence of the relative position toward the desired one.

For each agent, define the position and velocity tracking errors as
$e_{p_i}=p_i-p_i^d$ and  $e_{v_i}=v_i-\dot p_i^d$. The thrust input for each agent is chosen as
\begin{equation}
f_i=
\Big(
-K_{p_i}e_{p_i}
-K_{v_i}e_{v_i}
+m_i g e_3
-R_C T_i
+m_i\ddot p_i^d
\Big)\cdot R_i e_3,
\label{eq:thrust_control_law}
\end{equation}
where $K_{p_i},K_{v_i}\in\mathbb{R}^{3\times 3}$ are symmetric positive
definite gain matrices.

The control torques are designed using a standard geometric controller on
$SO(3)$,
\begin{align}
\tau_i
=&
-k_R e_{R_i}
-k_\Omega e_{\Omega_i}
+\Omega_i\times J_i\Omega_i
\nonumber\\&-J_i
\Big(
\hat\Omega_i R_i^\top R_i^d \Omega_i^d
-
R_i^\top R_i^d \dot\Omega_i^d
\Big),
\label{eq:torque_control_law}
\end{align}
where $k_R,k_\Omega>0$ are control gains.

The attitude error is defined by
$e_{R_i}
=
\frac12
\Big(
(R_i^d)^\top R_i
-
R_i^\top R_i^d
\Big)^\vee$, and the angular velocity error is given by
$e_{\Omega_i}
=
\Omega_i-R_i^\top R_i^d \Omega_i^d$. The desired attitude $R_i^d$ is chosen so that its third body axis aligns
with the desired total force direction. More precisely, let
\begin{equation}
b_{3_i}^d
=
\frac{
-K_{p_i}e_{p_i}
-K_{v_i}e_{v_i}
+m_i g e_3
-R_C T_i
+m_i\ddot p_i^d
}{
\left\|
-K_{p_i}e_{p_i}
-K_{v_i}e_{v_i}
+m_i g e_3
-R_C T_i
+m_i\ddot p_i^d
\right\|
},
\label{eq:desired_b3}
\end{equation}
and let $b_{1_i}^d$ denote a prescribed heading direction for agent~$i$.
Then the desired attitude is defined by
\begin{equation}
R_i^d=
\big[
b_{1_i}^d,\;
b_{3_i}^d\times b_{1_i}^d,\;
b_{3_i}^d
\big].
\label{eq:desired_attitude}
\end{equation}
We assume throughout that the denominator in \eqref{eq:desired_b3} is
nonzero.

Under the control laws \eqref{eq:thrust_control_law} and
\eqref{eq:torque_control_law}, standard results on geometric tracking
control for quadrotors on $SE(3)$ (see \cite{lee2010geometric}),  imply that, in the nominal case,
the position, velocity, attitude, and angular velocity tracking errors
decay to zero. The additional term $R_C T_i$ acts as a known feedforward
force induced by the catenary and does not alter the structure of the
tracking controller. Therefore, in the nominal catenary model, the proposed control law ensures
asymptotic convergence of the relative configuration of the catenary robot
to the desired geometry, while respecting the geometric constraints
imposed by the cable. In particular, since the cable-induced forces depend
only on the geometric configuration of the UAV pair, the controller
implicitly shapes the catenary through the regulation of the individual
agent trajectories.

\subsection{Robustness with respect to catenary model uncertainty}

We now analyze the robustness of the proposed geometric control law with
respect to uncertainty in the catenary model. 

In practice, the cable may
deviate from the ideal quasi-static catenary assumption due to
aerodynamic disturbances, elasticity, unmodeled mass distribution, or
transient dynamics. We model these effects as additive perturbations in
the catenary-induced forces. More precisely, we assume that the actual force transmitted by the cable
to agent~$i$ satisfies
\begin{equation}
R_C T_i^{\mathrm{act}} = R_C T_i + d_i,\quad i=1,2
\label{eq:actual_catenary_force}
\end{equation}
where $T_i$ denotes the nominal catenary tension from Section~III and
$d_i\in\mathbb{R}^3$ is an unknown disturbance term. We assume that
$\|d_i(t)\|\le \bar d$ for all $t\ge 0$ and for some finite constant
$\bar d>0$.

For the robustness analysis, we work with the relative translational
tracking errors $e_p = (p_2-p_1)-p_{12}^d$ and $e_v = (v_2-v_1)-\dot p_{12}^d$. We denote by \(K_p,K_v\in\mathbb{R}^{3\times 3}\) the 
relative position and velocity gain matrices, respectively, and assume
that they are symmetric positive definite. These gains are induced by
the tracking gains in \eqref{eq:thrust_control_law} after passing to relative coordinates.

Under the presence of cable model uncertainty, the relative translational
error dynamics become
\begin{equation}
m_r \dot e_v = -K_p e_p - K_v e_v + \Delta_d,
\label{eq:perturbed_relative_dynamics}
\end{equation}
where $m_r$ is the reduced mass and
\begin{equation}
\Delta_d=
\frac{m_2 d_2 - m_1 d_1}{m_1+m_2}
\label{eq:effective_disturbance}
\end{equation}
denotes the effective disturbance acting on the relative translational
dynamics. By assumption, there exists $\bar\Delta>0$ such that
$\|\Delta_d\|\le \bar\Delta$.

\begin{theorem}
Consider the catenary robot under the geometric control law
\eqref{eq:thrust_control_law}--\eqref{eq:desired_attitude}, and assume
that the nominal catenary-induced forces are perturbed according to
\eqref{eq:actual_catenary_force}. Suppose that the desired trajectories are smooth and bounded, and that the
attitude errors remain in a neighborhood of the desired attitudes where the quantity \(\operatorname{tr}(I-(R_i^d)^\top R_i)\) is locally positive
definite with respect to \(e_{R_i}\), for \(i=1,2\).

Then there exist sufficiently small constants
\(\alpha>0\) and \(\beta_i>0\), \(i=1,2\), such that the Lyapunov function \begin{align}
V
&=
\frac12 m_r e_v^\top e_v
+\frac12 e_p^\top K_p e_p
+\alpha e_p^\top e_v \label{eq:augmented_Lyapunov_candidate}        \\                
&+
\sum_{i=1}^2
\left[
\frac12 e_{\Omega_i}^\top J_i e_{\Omega_i}
+\frac{k_R}{2}
\operatorname{tr}\!\left(I-(R_i^d)^\top R_i\right)
+\beta_i e_{R_i}^\top J_i e_{\Omega_i}
\right]\nonumber
\end{align} is positive definite w.r.t 
\(e=(e_p,e_v,e_{R_1},e_{\Omega_1},e_{R_2},e_{\Omega_2})\). Moreover, there
exist constants \(c_p,c_v,c_R,c_\Omega,\gamma>0\) such that
\begin{equation}
\begin{aligned}
\dot V
\le&
-c_p\|e_p\|^2
-c_v\|e_v\|^2
-c_R\sum_{i=1}^2\|e_{R_i}\|^2        \\
&-
c_\Omega\sum_{i=1}^2\|e_{\Omega_i}\|^2
+\gamma\|\Delta_d\|^2 .
\end{aligned}
\label{eq:strict_ISS_dissipation_inequality}
\end{equation}
Consequently, the closed-loop tracking error system is input-to-state
stable (ISS) w.r.t the catenary force uncertainty \(\Delta_d\).
\end{theorem}

\begin{proof}
We first consider the relative translational dynamics
\[
m_r\dot e_v=-K_p e_p-K_v e_v+\Delta_d,\qquad \dot e_p=e_v .
\]
The translational part of \(V\) is
\[
V_t=
\frac12 m_r e_v^\top e_v
+\frac12 e_p^\top K_p e_p
+\alpha e_p^\top e_v .
\]
For \(\alpha>0\) sufficiently small, \(V_t\) is positive definite since
\(m_r>0\) and \(K_p=K_p^\top>0\). Differentiating \(V_t\) along the
relative error dynamics gives
\begin{align}
\dot V_t
=&
e_v^\top(-K_p e_p-K_v e_v+\Delta_d)
+e_p^\top K_p e_v+\alpha\|e_v\|^2    \nonumber \\
&+
\frac{\alpha}{m_r}e_p^\top
(-K_p e_p-K_v e_v+\Delta_d).
\label{eq:Vt_derivative_expanded}
\end{align}
Since \(K_p\) is symmetric, the terms
\(-e_v^\top K_p e_p\) and \(e_p^\top K_p e_v\) cancel. Hence
\begin{align}
\dot V_t
=&
-e_v^\top K_v e_v
+\alpha\|e_v\|^2
-\frac{\alpha}{m_r}e_p^\top K_p e_p
-\frac{\alpha}{m_r}e_p^\top K_v e_v        \nonumber \\
&+
\left(e_v+\frac{\alpha}{m_r}e_p\right)^\top\Delta_d .
\label{eq:Vt_derivative}
\end{align}
The mixed term is bounded, for any \(\eta_t>0\), by
\[
\frac{\alpha}{m_r}|e_p^\top K_v e_v|
\le
\frac{\alpha\|K_v\|}{2m_r}
\left(
\eta_t\|e_p\|^2+\frac{1}{\eta_t}\|e_v\|^2
\right),
\]
where $\|K_v\|$ is any matrix norm and in the last step we have applied Young's inequality. Thus, choosing \(\alpha>0\) sufficiently small, there exist
\(a_p,a_v>0\) such that
\[
\dot V_t
\le
-a_p\|e_p\|^2-a_v\|e_v\|^2
+
\left(e_v+\frac{\alpha}{m_r}e_p\right)^\top\Delta_d .
\]
Applying again Young's inequality to the disturbance term, and taking the
corresponding parameter sufficiently small, gives
\begin{equation}
\dot V_t
\le
-\bar a_p\|e_p\|^2
-\bar a_v\|e_v\|^2
+\gamma_t\|\Delta_d\|^2
\label{eq:translational_ISS_bound}
\end{equation}
for some \(\bar a_p,\bar a_v,\gamma_t>0\).

We now consider the attitude part. The rotational part of $V$ is composed of the terms
$$V_{R_{i}}=\frac12 e_{\Omega_i}^\top J_i e_{\Omega_i}
+\frac{k_R}{2}
\operatorname{tr}\!\left(I-(R_i^d)^\top R_i\right)
+\beta_i e_{R_i}^\top J_i e_{\Omega_i}$$
Under the geometric attitude controller, the feedforward terms cancel the
desired attitude dynamics and the closed-loop attitude error dynamics take
the form
\[
J_i\dot e_{\Omega_i}
=
-k_R e_{R_i}
-k_\Omega e_{\Omega_i},
\qquad i=1,2.
\]
Moreover, along the closed-loop attitude dynamics,
\begin{equation}
\frac{d}{dt}
\left[
\frac12 e_{\Omega_i}^\top J_i e_{\Omega_i}
+
\frac{k_R}{2}
\operatorname{tr}\!\left(I-(R_i^d)^\top R_i\right)
\right]
=
-k_\Omega\|e_{\Omega_i}\|^2 .
\label{eq:standard_attitude_energy_derivative}
\end{equation}
The derivative of the cross term satisfies
\begin{align}
\frac{d}{dt}
\left(\beta_i e_{R_i}^\top J_i e_{\Omega_i}\right)
=&
\beta_i \dot e_{R_i}^\top J_i e_{\Omega_i}
-\beta_i k_R\|e_{R_i}\|^2      \nonumber \\
&-
\beta_i k_\Omega e_{R_i}^\top e_{\Omega_i}.
\label{eq:attitude_cross_derivative}
\end{align}
In a neighborhood of the desired attitude, there exists \(c_i>0\) such that
\(\|\dot e_{R_i}\|\le c_i\|e_{\Omega_i}\|\). Therefore,
\[
\beta_i \dot e_{R_i}^\top J_i e_{\Omega_i}
\le
\beta_i c_i\|J_i\|\,\|e_{\Omega_i}\|^2,
\]
and Young's inequality gives
\[
\beta_i k_\Omega |e_{R_i}^\top e_{\Omega_i}|
\le
\frac{\beta_i k_R}{2}\|e_{R_i}\|^2
+
\frac{\beta_i k_\Omega^2}{2k_R}\|e_{\Omega_i}\|^2 .
\]
Combining these bounds with \eqref{eq:standard_attitude_energy_derivative}
yields
\begin{equation}
\begin{aligned}
\dot V_{R_i}
\le&
-\frac{\beta_i k_R}{2}\|e_{R_i}\|^2        \\
&-
\left(
k_\Omega
-\beta_i c_i\|J_i\|
-\frac{\beta_i k_\Omega^2}{2k_R}
\right)
\|e_{\Omega_i}\|^2.
\end{aligned}
\label{eq:attitude_strict_bound}
\end{equation}
Choosing each \(\beta_i>0\) sufficiently small, there exist
\(\bar c_{R_i},\bar c_{\Omega_i}>0\) such that
\begin{equation}
\dot V_{R_i}
\le
-\bar c_{R_i}\|e_{R_i}\|^2
-\bar c_{\Omega_i}\|e_{\Omega_i}\|^2 .
\label{eq:attitude_ISS_bound}
\end{equation}

Adding \eqref{eq:translational_ISS_bound} and
\eqref{eq:attitude_ISS_bound} for \(i=1,2\), we obtain
\eqref{eq:strict_ISS_dissipation_inequality} for suitable positive
constants \(c_p,c_v,c_R,c_\Omega,\gamma\).

Finally, since \(V\) is positive definite and locally equivalent to
\(\|e\|^2\), there exist \(m_1,m_2>0\) such that
\(m_1\|e\|^2\le V(e)\le m_2\|e\|^2\) in the considered neighborhood.
Let \(c_*=\min\{c_p,c_v,c_R,c_\Omega\}\). Since
\(V(e)\le m_2\|e\|^2\), \eqref{eq:strict_ISS_dissipation_inequality}
implies
\[
\dot V\le -\frac{c_*}{m_2}V+\gamma\|\Delta_d\|^2.
\]
Setting \(\lambda=c_*/m_2\) gives the standard ISS estimate $\dot V \le -\lambda V + \gamma\|\Delta_d\|^2$, and therefore the closed-loop tracking
error system is ISS with respect to \(\Delta_d\).
\end{proof}

\begin{corollary}
Under the assumptions of Theorem~1, let $\|\Delta_d\|_\infty
:=
\sup_{t\ge 0}\|\Delta_d(t)\|$. If \(\|\Delta_d\|_\infty<\infty\), then the tracking error is uniformly
ultimately bounded. More precisely, there exists a constant \(c>0\) such
that
\[
\limsup_{t\to\infty}\|e(t)\|
\le
c\,\|\Delta_d\|_\infty .
\]
In particular, if \(\Delta_d\equiv 0\), then the tracking errors converge
asymptotically to the origin. Moreover, the ultimate bound decreases as translational and attitude damping gains increase.
\end{corollary}

\begin{proof}
From Theorem~1, there exist positive constants
\(m_1,m_2,\lambda,\gamma\) such that, in the considered neighborhood,
$m_1\|e\|^2\le V(e)\le m_2\|e\|^2$ and $\dot V
\le
-\lambda V+\gamma\|\Delta_d\|^2$. If \(\Delta_d\equiv 0\), then
$\dot V\le -\lambda V$, and therefore \(V(t)\le V(0)e^{-\lambda t}\). Since \(V\) is locally
equivalent to \(\|e\|^2\), the tracking errors converge exponentially,
and hence asymptotically, to the origin.

If \(\|\Delta_d\|_\infty<\infty\), then $\dot V
\le
-\lambda V+\gamma\|\Delta_d\|_\infty^2$. Solving this scalar differential inequality yields
\[
V(t)
\le
e^{-\lambda t}V(0)
+
\frac{\gamma}{\lambda}
\left(1-e^{-\lambda t}\right)
\|\Delta_d\|_\infty^2 .
\]
Taking the limit superior as \(t\to\infty\) gives
\[
\limsup_{t\to\infty}V(t)
\le
\frac{\gamma}{\lambda}\|\Delta_d\|_\infty^2 .
\]
Using \(m_1\|e(t)\|^2\le V(t)\), we obtain
\[
\limsup_{t\to\infty}\|e(t)\|
\le
\sqrt{\frac{\gamma}{m_1\lambda}}\,
\|\Delta_d\|_\infty .
\]
Thus the result holds with $c=\sqrt{\frac{\gamma}{m_1\lambda}}$. Since \(\lambda\) increases with the
translational and attitude damping gains, the ultimate bound decreases
as these gains increase.
\end{proof}

\subsection{Numerical example}

Next we illustrate the theoretical results through the closed-loop error
dynamics used in the proof of Theorem~1. We consider a constant desired
relative position
\(p_{12}^d=\begin{bmatrix} d_0 & 0 & 0 \end{bmatrix}^\top\),
with \(0<d_0<l\), so that the desired catenary configuration is feasible.
UAV~1 is stabilized at a fixed point \(p_{1,0}\), while UAV~2 tracks the
induced reference \(p_2^d(t)=p_{1,0}+p_{12}^d\). The reduced mass is set to
\(m_r=1\), the relative position gain is \(K_p=6I_3\), and the nominal
damping gains are \(K_v=7I_3\) and \(k_\Omega=6\). The attitude gains and
inertia matrices are chosen as \(k_R=8\),
\(J_1=J_2=\operatorname{diag}(0.08,0.08,0.12)\). The small Lyapunov cross
terms are taken as \(\alpha=0.08\) and \(\beta_1=\beta_2=0.02\). The initial
condition is selected in a neighborhood of the desired configuration.

To test robustness, we introduce a bounded perturbation in the relative
translational dynamics of the form
\[
\Delta_d(t)=\bar\Delta
\begin{bmatrix}
\sin(\omega_1 t) \\
0.5\cos(\omega_2 t) \\
0.3\sin(\omega_3 t)
\end{bmatrix},
\]
with \(\omega_1=1.3\), \(\omega_2=0.9\), and \(\omega_3=1.7\). The scalar
\(\bar\Delta>0\) determines the disturbance amplitude. In
Fig.~\ref{fig:theorem1}, we use \(\bar\Delta=0.35\) for the perturbed case,
while the nominal case corresponds to \(\bar\Delta=0\).

Fig.~\ref{fig:theorem1} illustrates Theorem~1. In the nominal case, the
tracking error converges to zero. Under bounded perturbations, the error
remains ultimately bounded and converges to a small neighborhood of the
origin, consistently with the ISS estimate.

Fig.~\ref{fig:corollary1} illustrates Corollary~1. The left panel shows
that the ultimate tracking error grows approximately linearly with
\(\|\Delta_d\|_\infty\), comparing low damping
\((K_v=4I_3,k_\Omega=3.5)\) and high damping
\((K_v=10I_3,k_\Omega=8)\). The right panel fixes \(\bar\Delta=0.40\) and
shows that increasing the translational damping gain \(K_v=k_vI_3\)
reduces the ultimate tracking error.

\begin{figure}[h!]
\centering
\includegraphics[width=\linewidth]{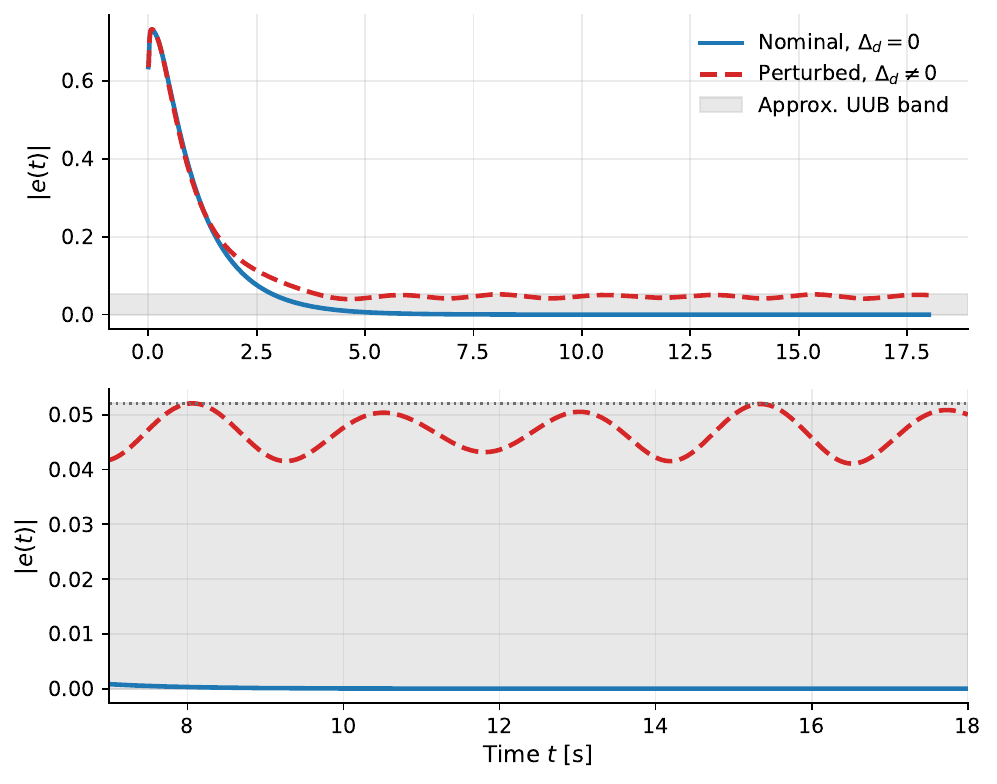}
\caption{The nominal error converges to zero, while bounded perturbations produce an ultimately bounded response. The gray band indicates the numerical tail bound.}
\label{fig:theorem1}
\end{figure}

\begin{figure}[h!]
\centering
\includegraphics[width=\linewidth]{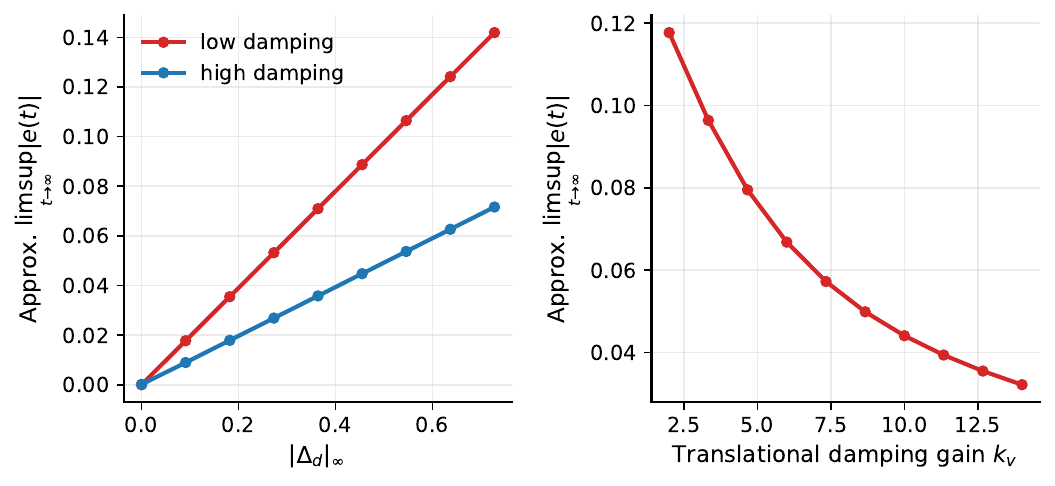}
\caption{Left: ultimate error versus \(\|\Delta_d\|_\infty\) for low and high damping. Right: ultimate error versus \(k_v\) for \(K_v=k_vI_3\) and fixed disturbance magnitude.}
\label{fig:corollary1}
\end{figure}

\section{Conclusions}

We developed a robust geometric control framework for catenary robots on
\(SE(3)\), treating the cable as a configuration-induced geometric subsystem.
The proposed controller compensates nominal catenary forces and achieves
local input-to-state stability with respect to unstructured force
uncertainties, with an explicit ultimate bound for bounded perturbations.
Future work will extend the model to dynamic and elastic cables, contact-aware
tasks, and experimental validation on multi-UAV platforms.

%\FloatBarrier
\bibliographystyle{IEEEtran}
\bibliography{cluster}

\end{document}